\documentclass[letterpaper, 10 pt, conference]{ieeeconf}  

\IEEEoverridecommandlockouts                              

\usepackage{graphicx} 
\usepackage{mathptmx} 
\usepackage{times} 
\usepackage{amsmath} 
\usepackage{amssymb}  
\usepackage{color}

\newtheorem{theorem}{Theorem}
\newtheorem{definition}{Definition}
\newtheorem{lemma}{Lemma}

\newcommand{\sign}{\mathrm{sign}}

\newcommand{\seq}{\mbox{$\! \, = \, \!$}}

\title{\LARGE \bf
Lyapunov stability analysis of rigid body systems with multiple contacts by sums-of-squares programming}

\author{P\'eter~L.~V\'arkonyi
\thanks{*This work was supported by the National Research, Innovation and Development Office of Hungary under grant $K124002$.}
\thanks{P. L. V\'arkonyi is with the Department of Mechanics Materials and Structures, Budapest University of Technology and Economics, H-1111 Budapest, Hungary
        {\tt\small varkonyi.peter@epk.bme.hu}}%
}

\begin{document}

\maketitle
\thispagestyle{empty}
\pagestyle{empty}

\begin{abstract}

Reliable quasi-static object manuipulation and robotic locomotion require verification of the stability of equilibria under rigid contacts and friction. In a recent paper, M. Posa, M. Tobenkin, and R. Tedrake demonstrated that
sums-of-squares (SOS) programming can be used to verify Lyapunov stability via Lyapunov's direct method. This test was successfully applied to several simple problems with a single point contact. At the same time it has been found that this method is too conservative for several multi-contact systems. In this paper, an extension of Lyapunov's direct method is proposed,  which makes use of several Lyapunov functions, and which allows \emph{temporary} increase of those Lyapunov function along a motion trajectory. The proposed method remains compatible with SOS programming techniques. The improved stability test is successfully applied to a rigid body with 2 point contacts, for which the exact conditions of Lyapunov stability are unknown.

\end{abstract}

\section{INTRODUCTION}

Many tasks in robotics involve unilateral contact with friction between hard objects. Small perturbations of an equilibrium state may induce stick-slip transitions, as well as contact separation and impacts \cite{brogliato1999nonsmooth,leine2010historical}. Within the framework of rigid body dynamics,  non-smoothness, and discontinuity of the response prevents one from using classical tools of stability analysis like linearization. At the same time, dry friction induces continuous sets of equilibrium states \cite{leine2008stability}, which means that stable individual points within a set are not associated with local minima of potential energy. Moreover, the emerging hybrid dynamics has special features like Zeno sequences of impact events \cite{zhang2001zeno,or2010stability}, and issues of non-existence and non-uniqueness \cite{champneys2016painleve}. As a result, general methods to test the Lyapunov stability of equilibria under contact and friction remain unavailable \cite{liapunov2016stability,orlov2008discontinuous}. 

Stability analysis of a planar model with one single point contact is straightforward: an equilibrium is stable under any type of small state perturbation if it is stable under rolling motion, and additionally the Newtonian coefficient of restitution of light impacts in the neighborhood of that state is below 1. In contrast, stability analysis of models with multiple contact points becomes an open problem. In the case of 2 contact points and ideally inelastic impacts, an almost exact semi-analytic condition has been derived \cite{varkonyi2017lyapunov,or2021experimental}. Some highly conservative conditions have been found in the case of partially elastic impacts \cite{or2008hybrid,varkonyi2012lyapunov}. The last work used a custom-made Lyapunov function for the verification of stability.  Systems with more than 2 contact points, models with extended areas of contact, as well as 3D models remain unexplored.

A novel algorithmic approach to stability analysis by Posa et al. \cite{posa2015stability} uses sums-of-squares polynomials and semi-definite pogramming  to develop sufficient conditions of local Lyapunov stability of various systems. In addition, their method delivers  estimations of the basin of attraction of a stable point, and it can also be used to verify positive invariance of a set of states. This approach is quite general as it is formally  applicable to any system for which the equations of motion and the constraints are expressed as polynomial equations or inequalities. At the same time, success not guaranteed as the proposed conditions are conservative. 
Notably, all of the models successfully tested by \cite{posa2015stability} have only one point of contact, and it remains an open question if more challenging problems can be addressed by this method.

In this paper we focus on Lyapunov stability, i.e. resilience against local state perturbations. We find that the method proposed by \cite{posa2015stability} fails to verify stability of an important class of model problems: a planar model of a rigid body with 2 point contacts. This finding motivates the development of several improvements of the stability theory. Most importantly, the semi-definite programming techniques of \cite{posa2015stability} are combined with an extension of Lyapunov's direct method inspired by prior work of the author \cite{varkonyi2012lyapunov}. In its original form, Lyapunov's direct method uses Lyapunov functions defined over the state space of the system, which are non-increasing along motion trajectories. The proposed extension makes use of Lyapunov functions which may temporarily increase along trajectories, provided that they have an decreasing overall trend. In addition the theory is also extended here to multiple polynomial Lyapunov functions, which improves the applicability of semi-definite programming techniques. Finally, the present work makes use of a standard approximation of piecewise continuous systems: the so-called zero-order dynamics (ZOD). In each mode of motion, the equations of motion are approximated by their lowest-order (constant) terms, which gives accurate local description of the system in a small neighborhood of an equilibrium state. This approach is highly similar to linearization techniques of smooth dynamical systems.

In Section~\ref{sec:problem},  general notation, and problem statement are introduced, and key results of reference \cite{posa2015stability} are reviewed. This is followed by proposed extensions of existing stability theory in Section \ref{sec:extended}. Then, a family of test problems is introduced and analyzed (Section \ref{sec:example}). The paper is closed by Conclusions and by pointing out related open problems.

\section{PROBLEM STATEMENT AND PREVIOUS RESULTS}
\label{sec:problem}


\subsection{Problem statement}

We consider a rigid body or a  rigid multibody system in 2 dimensions with $c\geq 1$ unilateral point contacts and $n$ degrees of freedom. The state $x$ of the system is given by the state vector $x=(q,v)$ composed of generalized coordinates $q\in\mathbb{R}^n$ and generalized velocities $v=\dot q$ where dot means derivation with respect to time. Hence we have a $2n$-dimensional state space $x\in\mathcal{S}\equiv\mathbb{R}^{2n}$. We will consider an equilibrium state given by $x_0=(q_0,0)$. Without loss of generality $q_0=0$ is assumed. 

We use lower indices to denote elements of vectors, and columns of matrices. For example $\gamma_i$ is the $i^{th}$ element of vector $\gamma$, and $M_i$ is the $i^{th}$ column of matrix $M$. Furthermore $J_\gamma(q)$ means the Jacobian of $\gamma(q)$ if $\gamma$ is a vector or the gradient of $\gamma(q)$ if $\gamma$ is a scalar.  

The admissible set $\mathcal{A}$ is defined as a subset of the state space given by $c$ scalar inequality constraints
\begin{align}
\gamma_i(q)\geq 0; \; i\in\{1,2,...,c\}
\label{eq:gap}
\end{align}
where $\gamma_i(q)$ are \emph{gap functions} associated with the unilateral contacts. As we perform local analysis of the equilibrium $q_0=0$, we are only interested in constraints with $\gamma_i(0)=0$. 

The unilateral contacts may give rise to non-negative normal contact forces $\lambda_{iN}$ subject to the linear complementarity condition
\begin{align}
\lambda_{iN}\cdot \gamma_i = 0 \leq \lambda_{iN} , \gamma_i 
\end{align}
furthermore if $\gamma_i=0$ then we also have an analogous condition at the velocity level:
\begin{align}
\lambda_{iN}\cdot \dot \gamma_i = 0 \leq \lambda_{iN} , \dot \gamma_i
\end{align}
In the presence of friction, it is convenient to introduce tangental displacement functions $\sigma(q)$ associated with each contact, since the behaviors of contacts depend on the relative tangential velocities given by $\dot{\sigma}(q,v)$. The sign of $\dot\sigma_i$ may be positive, zero, or negative, which correspond to slip in the positive direction, stick, and negative slip at contact $i$. Coulomb friction force is subject to the constraint
\begin{align}
\lambda_{iT}\in -\sign(\dot \sigma_i)\mu_i\lambda_{iN}\; i\in\{1,2,...,c\}
\label{eq:v}
\end{align}
where $\mu_i\geq 0$ are friction coefficients associated with the contacts and $\sign(\cdot)$ denotes the set-valued sign function. 

The dynamics of the system is governed by the manipulator equations
\begin{align}
\dot q &= v\\
H(q)\dot v+C(q,v) & = J_N(q)\lambda_N+J_T(q)\lambda_T
\label{eq: manipulator}
\end{align}
where the mass matrix $H$ depends on the current configuration $q$ of the system. The term $C$ includes the effect of free forces as well as gyroscopic terms. The two terms on the right-hand side represent the effects of normal and tangential contact forces. If $H$ is invertible, one can express the instanataneous acceleration $\dot v$ in terms of the (unknown) contact forces. 

As an example, consider a planar rigid body with two sharp vertices resting on two straight surfaces as in Fig. \ref{fig:body}. We fix a global coordinate frame such that the origin conincides with the center of mass in the initial equilibrium configuration and the $x$ axis is parallel to the line through the initial positions of the contact points.  The object is subject to constant external forces, which are lumped into a resultant force of size $F$ and directional angle $\alpha$ acting at the center of mass as well as a resultant torque $T$. Without loss of generality we assume that the mass, the radius of gyration, and the gravitational constant are equal to 1. The remaining model parameters include the lengths $h$, $l_1$, $l_2$, the angles $\phi_1$, $\phi_2$ and the friction coefficients $\mu_1$, $\mu_2$. We can use the global coordinates $(x,z)$ of the center of mass, and the rotation angle $\theta$ of the body as generalized coordinates, i.e. $q=(x,z,\theta)$. The parameters of the manipulator equation and the gap functions and tangential displacement functions are:
\begin{align}
H=
\left[
\begin{array}{ccc}
1&0&0\\
0&1&0\\
0&0&1
\end{array}
\right ] ,
J_N=
\left[
\begin{array}{cc}
-\sin\phi_1&-\sin\phi_2\\
\cos\phi_1&\cos\phi_2\\
\zeta_1&\zeta_2
\end{array}
\right ]
\label{eq:H}\\
C=
\left[
\begin{array}{c}
F\sin\alpha\\
-F\cos\alpha\\
T
\end{array}
\right ],
J_T=
\left[
\begin{array}{cc}
\cos\phi_1&\cos\phi_2\\
\sin\phi_1&\sin\phi_2\\
\psi_1& \psi_2
\end{array}
\right ] \label{eq:C}\\
\gamma_i(q)=
-y\sin\phi_i+z\cos\phi_i - \xi_i+\psi_i
\label{eq:gammabiped}\\
\sigma_i(q)=y\cos\phi_i+z\sin\phi_i-\eta_i+\zeta_i  
\label{eq:sigmabiped}
\end{align}
where the following notation has been used:
\begin{align}
\eta_i&=l_i\cos\phi_i-h\sin\phi_i \label{eq:eta}\\
\xi_i&=-h\cos\phi_i-l_i\sin\phi_i\label{eq:xi}\\
\zeta_i&=l_i\cos(\phi_i-\theta)-h\sin(\phi_i-\theta)\\
\psi_i&=-h\cos(\phi_i-\theta)-l_i\sin(\phi_i-\theta)
\label{eq:psi}
\end{align}
\begin{figure}[h]
\centering
\includegraphics[width=\columnwidth]{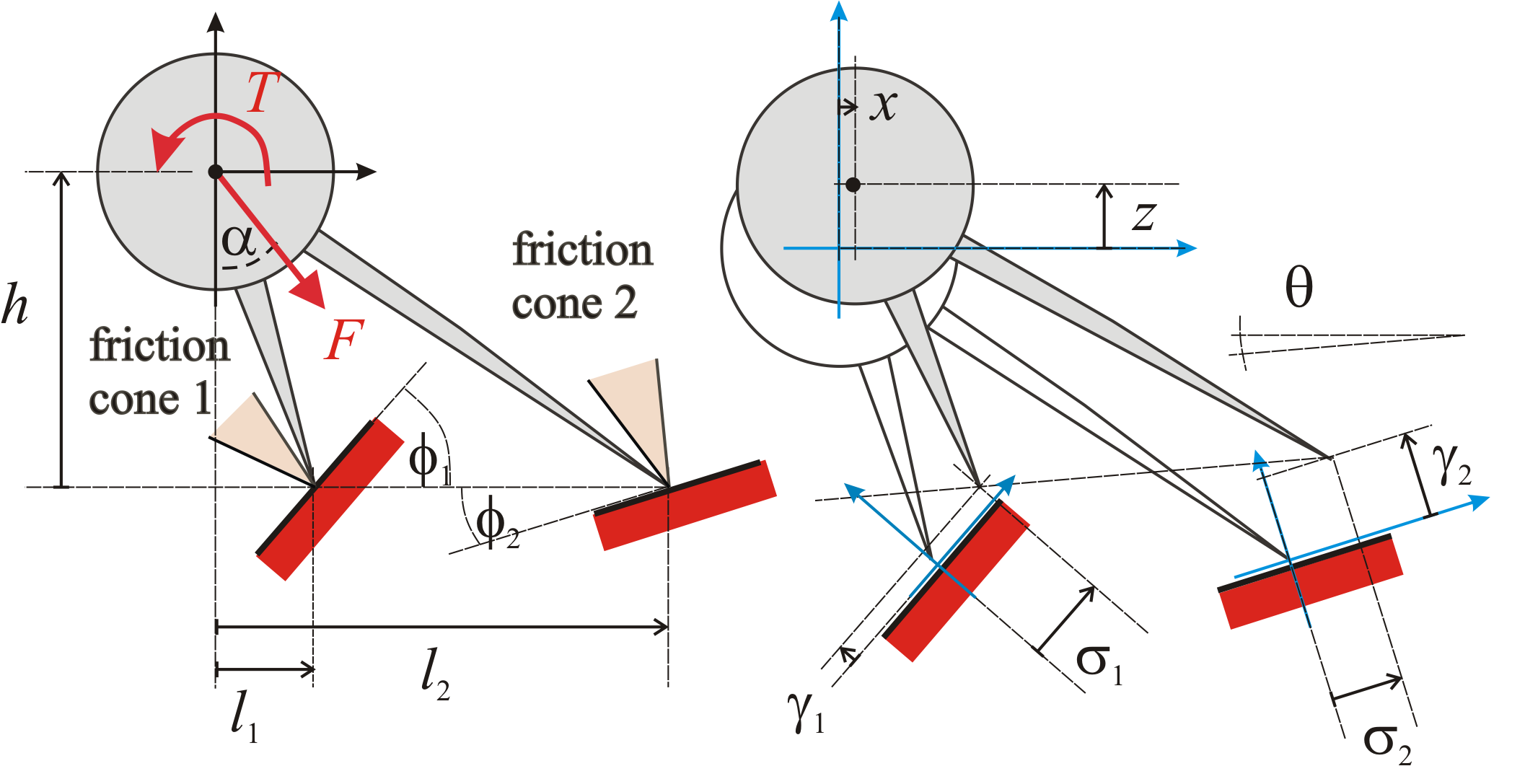}
\caption{A planar rigid body on two point contacts. Left: equilibrium configuration with model parameters. Right: general configuration with state variables, gap functions and tangential displacement functions.}
\label{fig:body}
\end{figure}

\subsection{The hybrid dynamics approach}
In order to find contact forces, and acceleration simultaneously, several standard methods can be used. Numerical methods may treat the problem as a linear complementarity problem, which adresses all cases of the sign function in a unified framework. In contrast, detailed analysis of the emerging motion is usually done within the framework of hybrid dynamics (especially for systems with moderate number of equilibria).
The hybrid dynamics approach is used here, i.e. it is assumed that the system undergoes episodes of continuous motion in one of its \emph{contact modes}. Such episodes are interrupted by \emph{contact mode transitions} and \emph{impacts}.

\subsection{Contact modes}

Contact modes are defined based on the signs of $\gamma_i$, $\dot\gamma_i$, and $\dot\sigma_i$. In particular each individual contact point of a planar model is in one of the following states: free flight (F), stick (S), and slip in either one of two directions (P,N), see Table \ref{tab.modes}. Thus the system has $4^c$ contact modes. Each contact point and contact mode has a set of \emph{kinematic admissibility contraints}, which are equalities or inequalities in the state variables. These constraints reduce the number of contact modes to be considered in a given state of the system. For example, point $i$ cannot be in slip or stick state if $\gamma_i>0$. 

In addition, each point delivers exactly two equality constraints involving  contact forces ($\lambda_{N,i}$, $\lambda_{T,i}$) and/or accelerations $\dot v$, which are combined with the manipulator equation in order to determine the instantaneous values of the contact forces. Finally, a second set of inequalitiy constraints  involves contact forces and accelerations. Those \emph{consistency constraints} are tested after the contact forces and accelerations have been determined.

We consider an \emph{equilibrium state}, i.e. a static state $x_0=(0,0)$, in which the contact mode with all contact points in stick state yields a consistent solution. We also restrict our attention to those systems in which kinematic constraints ensure immobility if all contacts are in S state. This is true for the example of Figure \ref{fig:body}, however not true if the same object has only one point contact.

An important limitation of the hybrid dynamics approach is the non-uniqueness and non-existence of consistent contact modes in some systems \cite{stewart2000rigid}. We will distinguish between two types of this situation. The first one is given by .
\begin{definition}\emph{
An equilibrium state $x_0$ is called \textbf{ambiguous}, if the system has (in addition to the sustained equilibrium state) a non-static consistent contact mode as well.
}\end{definition} 
Ambiguity is quite common and it has been proven that ambiguous equilibria are never Lyapunov-stable \cite{or2008hybrid-I}. Conditions of ambiguity for the previous example were also given by \cite{or2008hybrid-I}.

The second type involves solution non-existence as well as all forms of non-uniqueness in non-static states. It is often referred to as Painlev\'e's paradox \cite{champneys2016painleve}. Non-existence is resolvable by considering impulsive contact forces in non-impacting states. Non-uniqueness is in general not resolvable within the framework of rigid models. Nevertheless, with the approach that we develop in this paper it is possible to verify Lyapunov stability even when the hybrid dynamics approach fails to identify a unique solution.  

The second important limitation of the hybrid dynamics approach is complexity. The number of contact modes is exponential in $c$ which explains why this approach is limited to moderate values of $c$.

\begin{table*}
\centering{
\begin{tabular}{|c|l|l|l|l|}\hline
   Letter & contact mode & kinematic & equalities & consistency
   \\
   & & admissibility & & constraints \\   \hline
  {\bf S} & sticking & $\gamma_i = 0$ & $\ddot{\gamma}_i = 0$ & $|\lambda_{iT}| \leq \mu_i (\lambda_{iN})$ \\ 
  & &  $\dot\gamma_i \seq \dot\sigma_i \seq 0$ & $\ddot\sigma_i=0$ &  \\ \hline
{\bf F} & free & $\gamma_i \geq 0$ and & $\lambda_{iN}=\lambda_{iT} = 0$ & $\ddot{\gamma}_i > 0$ if $\gamma_i \seq \dot\gamma_i \seq 0$\\
& &$\dot\gamma_i \geq 0$ if $\gamma_i \seq 0$ & &  \\ \hline
{\bf P} & positive slip & $\gamma_i\seq \dot\gamma_i \seq 0$, & $\ddot{\gamma}_i =0$, & $\lambda_{iN} \geq 0 $,\\ & & $\dot\sigma_i \geq 0$ & $\lambda_{iT} =  -\mu (\lambda_{iN})  $& $\ddot\sigma_i > 0$ if $\dot\sigma_i \seq 0$\\ \hline
{\bf N} & negative slip & $\gamma_i\seq \dot\gamma_i \seq 0$, & $\ddot{\gamma}_i =0$, & $\lambda_{iN} \geq 0 $,\\ & & $\dot\sigma_i \leq 0$ & $\lambda_{iT} =  \mu (\lambda_{iN})  $& $\ddot\sigma_i < 0$ if $\dot\sigma_i \seq 0$\\ \hline
\end{tabular}
\caption{Contact modes of a single contact point}
\label{tab.modes}}

\end{table*}

\subsection{Contact mode transitions and impacts}
Continuous motion in a given contact mode is not possible unless the corresponding constraints are satisfied. If a consistency constraints or an admissibility constraint related to tangential velocities is violated, a \emph{contact mode transition} occurs. For example violation of $\lambda_{iN}\geq0$ triggers liftoff at point $i$ and violation of $\dot\sigma_i>0$ triggers slip-stick transition or slip reversal at point $i$. These transitions are marked by continuous but non-smooth velocity functions. 

In contrast, if any of the kinematic admissibility constraints related to  $\gamma_i$ and $\dot \gamma_i$ is violated (in other words, if the boundary of the admissible set $\mathcal{A}$ is reached), then an impact occurs. The hybrid dynamics approach requires \emph{rigid impact models}, which treat impacts as instantaneous velocity jumps generated by instantaneous impulses at the contact points. Hence velocity becomes a discontinuous, piecewise smooth function of time, with non-identical left and right limits $v(t^-)$, $v(t^+)$ at impact times. 

Hence, an impact is modeled by a map assigning to each possible pre-impact state $(q(t),v(t^-))$ with $\gamma_i=0$ and $\dot \gamma_i^-\leq 0$ for some $i$, an admissible post-impact state with updated velocity value $v(t^+)$ satisfying the condition $\dot \gamma_i\geq0$.

For an \emph{impact at a single point} $i$, many algebraic impact models exist. We will adopt the classical Whittaker-Kane-Levinson model \cite{kane1985dynamics}, which assumes that the Newtonian coefficient of restitution parameter $e$ is known in advance. Hence the normal component of the post-impact normal velocity is given by:
$
\dot\gamma_i(q(t),v(t^+))=-e\dot\gamma_i(q(t),v(t^-))
$
Given that $\gamma$ is a function of $q$, the previous relation is equivalently expressed as 
\begin{equation}
J_{\gamma,i}(q)^Tv(t^+)=
-e
J_{\gamma i}(q)^Tv(t^-).
\end{equation}

In addition, it is assumed that Coulomb's law is satisfied in the following form:
\begin{equation}
\Lambda_{iT}\in -\sign(\dot \sigma_i(t^+))\Lambda_{iN}\
\label{eq:Coulomb-impact}
\end{equation}
where $\Lambda_{Ni}$ and $\Lambda_{T,i}$ are signed normal and tangential impulses at point $i$.

The constraints outlined above can be combined with a discrete version of the manipulator equation
\begin{equation}
 H(q)(v(t^+)-v(t^+)) = J_{Ni}(q)\Lambda_{Ni}+J_{T,i}(q)\Lambda_{T,i}
 \label{eq:manipulator2}
\end{equation}
where $J_{N,i}$, $J_{T,i}$ are the $i^{th}$ column of $J_{N}$ and $J_{T}$, respectively. By combining these equations the impact map  can be expressed in closed form \cite{kane1987explicit}. In particular, an impact can be sticking ($\dot \sigma_i(t^+)=0$) or slipping in either direction ($\pm \dot \sigma_i(t^+)>0$) and it can be shown that exactly one of the 3 solution candidates satisfies the consistency conditions. 

The impulse of a sticking impact is then given by
\begin{align}
\Lambda_{N,i}&=\frac{ 
-(1+e)J_{\gamma,i}(q)^Tv(t^-)-J_{\gamma,i}(q)^TH^{-1}(q)J_{T,i}(q)A}
{J_{\gamma,i}(q)^TH^{-1}(q)(J_{N,i}+J_{T,i}B)}\\
\Lambda_{T,i}&=A+B\Lambda_{N,i}
\end{align}
with 
$$
A=\frac{J_{\sigma i}^Tv(t^-)}{J_{\sigma i}^TH^{-1}(q)J_{T,i}}
,
B=\frac{J_{\sigma i}^TH^{-1}(q)J_{N,i}}{J_{\sigma i}^TH^{-1}(q)J_{T,i}}.
$$
The impulse of a slipping impact in the $\pm$ direction is given by
\begin{align}
\Lambda_{N,i}&=\frac{ 
-(1+e)J_{\gamma,i}(q)^Tv(t^-)}
{J_{\gamma,i}(q)^TH^{-1}(q)(J_{N,i}\mp \mu_iJ_{T,i})}\\
\Lambda_{T,i}&=\mp\mu_i \Lambda_{N,i}
\end{align}

Simultaneous \emph{impacts at several contact points} can be modeled in several different ways. For example, iterative solution techniques of the numerical time-stepping methods of contact dynamics compute the final states of multi-point impacts by considering sequences of impacts with single points, where the subsequent impact point are chosen according to a priori rules \cite{acary2008numerical}. In this work, we adopt the idea that multi-point impacts are equivalent of some (finite or infinite) sequence of single-point impacts, but we do not make any particular assumption about the order of impact points in those sequences. This assumption allows us to restrict our attention to single-point impacts.

\subsection{Zero-order dynamics}\label{sec:ZOD}
In a small neighborhood of the equilibrium state $x_0=0$, the dynamics can be approximated by modified versions of equation \eqref{eq: manipulator},\eqref{eq:manipulator2} and of the constraints, in which the gap and tangential displacement functions are linearized around $q=0$ and all other state-dependent terms are approximated by their nominal values at $(q,v)=(0,0)$:
\begin{align}
\gamma(q)&\approx J_{\gamma}(0)q;\,\sigma(q)\approx J_{\sigma}(0)q\\
H(q)&\approx H(0);\,C(q,v)\approx C(0,0)\\
J_N(q)&\approx J_N(0);\,J_T(q)\approx J_T(0)
\end{align}
%

For the previously mentioned example of a planar biped, $H$, $C$ are state-independent. $J_N$, and $J_T$ and the functions \eqref{eq:gammabiped}-\eqref{eq:sigmabiped} are approximated by:
\begin{align}
J_N\approx 
\left[
\begin{array}{cc}
-\sin\phi_1&-\sin\phi_2\\
\cos\phi_1&\cos\phi_2\\
\eta_1&\eta_2
\end{array}
\right ] 
\label{eq:JbipedZod}\\
J_T\approx 
\left[
\begin{array}{cc}
\cos\phi_1&\cos\phi_2\\
\sin\phi_1&\sin\phi_2\\
\xi_1& \xi_2
\end{array}
\right ] \\
\gamma_i(q)\approx 
-y\sin\phi_i+z\cos\phi_i + \theta\eta_i \label{eq:gammabipedZOD}\\
\sigma_i(q)\approx y\cos\phi_i + z\sin\phi_i-\theta\xi_i  
\end{align}
where we used the previously introduced shorthand notations \eqref{eq:eta}, \eqref{eq:xi}.
This approximation will be referred to as zero-order dynamics or ZOD. The ZOD  delivers linear admissibility conditions, and constant accelerations and contact forces in each contact mode. On the one hand, the ZOD allows to solve the equations of motion in all contact modes in closed from, which will be exploited later. On the other hand, the consistency conditions of each contact mode related to accelerations are state-independent: they can be verified for each contact mode in a single step. The ZOD approximation is also applied to impact maps.

The ZOD approximation yields a close approximation of the real dynamics in a small neighborhood of state $x=0$. Nevertheless this approximation brakes down for systems with marginally consistent contact modes. For example, the ZOD predicts $\ddot\gamma_2 > \ddot\gamma_1 = 0$ for the planar biped in $FF$ mode if $\alpha=\pi/2$, $\phi_1=0$, and $-\pi < \phi_2 < 0$. Hence $FF$ is marginally consistent in the state $q=v=0$. According to the exact equations, the same system has $\ddot\gamma_1>0$ in each state with $q=0$ and $\dot\theta\neq 0$ if the geometric parameter $h$ is strictly positive. Hence the consistency of FF under the exact dynamics in a small neighborhood of $x=0$ cannot be decided based on the ZOD. 

A similar problem may occur in relation with marginally admissible contact modes. For example,consider the planar biped with $\phi_1=\pm\pi/2$, $l_1>l_2$, $h=0$ in a state such that point 2 is immobile, $\dot\theta > 0$, and  $\theta\neq 0$. Contact mode $PS$ is marginally kinematically admissible under the ZOD but the same mode under the exact dynamics is not admissible as $\gamma_1>0$ .

In what follows we exclude the marginal cases mentioned above, and we consider dynamics under the ZOD. It is not proven formally that local stability under ZOD implies local stability under the exact equations of motion. Nevertheless the numerical techniques used for stability verification are robust against small numerical error, hence it is plausible to assume that the (locally) small error introduced by the ZOD approximation does not invalidate the stability tests with the exception of the degenerate cases outlined above.

\subsection{Stability}
We will use the standard concept of \emph{Lyapunov stability}. For continuous-time systems, Lyapunov stability is defined as
\begin{definition}\emph{ Let $x_0=(q_0,0)$ be an equilibrium state.  This configuration is called \textbf{Lyapunov stable (LS)} if for every arbitrarily small $\epsilon > 0$ there exists $\delta(\epsilon)>0$ such that for any kinematically admissible initial state $x(0)$ that satisfies
$|x(0)-x_0|<\delta $,
 all states $x$ along the emerging motion trajectory satisfy $|x|<\epsilon$.
}\end{definition}

The trajectories of impacting (hybrid) systems are piecewise smooth, but discontinuous due to impacts. 
The definition of Lyapunov stability stated above is applicable to such systems.

Lyapunov stability is most often demonstrated with the help of Lyapunov's direct method: i.e. by constructing a scalar Lyapunov function $V(x)$ over the state space such that $V$ has a local minimumpoint at the equilibrium state under investigation, and $V$ is non-increasing along motion trajectories.  The adaptation of Lyapunov's direct method to impacting systems \cite{leine2007stability} is summarized below. 

Let $\mathcal{B}_\rho$ denote a ball of radius $\rho$ in state space centered around $x_0$. Furthermore, let $\textrm{cl}$ denote the closure of a set. For example, $\textrm{cl}(\mathcal{A})$ is the union of the admissible set and possible pre-impact states where $\dot\gamma_i<\gamma_i=0$ for some $i$. 
\begin{theorem}\emph{[\cite{leine2007stability}, Thm. 6.23] \label{thm:Lyapunov1} If there exist
\begin{enumerate}
\item a positive value $h$
\item a continuous, strictly increasing scalar function $\alpha$ with $\alpha(0)=0$
\item a continuously differentiable function $V : \mathbb{R}^{2n}\rightarrow \mathbb{R}$ such that for any initial state $x\in \mathcal{B}_h \cap \textrm{cl}(\mathcal{A})$ \begin{enumerate} 
\item $V(x)\geq\alpha(|x-x_0|)$
\item $V$ is non-increasing along any continuous piece of trajectory through $x$ and at jumps associated with impacts 
\end{enumerate}
\end{enumerate}
then $x_0$ is stable in the sense of Lyapunov.
}\end{theorem}
The proof of this statement is reviewed below, as its logical steps will be reused during the development of the extended stability theory.

\begin{proof}
Consider an arbitrary $\epsilon>0$. Define $c$ as $c=\alpha(\min(\epsilon,h))$. Then $V(x)\geq c$ along the boundary of $\mathcal{B}_{\min(h,\epsilon)}$ hence $V(x)$ has a closed, connected sublevel set  $\Omega_c=\lbrace x: V(x)\leq c\rbrace$ such that $\Omega_c$ contains $x_0$ in its interior, furthermore $\Omega_c\subset \mathcal{B}_{\min(h,\epsilon)}$. Due to condition 3b), $\Omega_c$ is positively invariant set of the dynamics (see \cite{leine2007stability}, Proposition 6.5).
Then we can find a scalar $\delta>0$ such that $\mathcal{B}_\delta\subset(\Omega_c\cup\mathcal{A}^c)$ where $^c$ means complement. Now if the initial point of a trajectory is in $\mathcal{B}_\delta$, then the positive invariance of $\Omega_c$ implies that the trajectory remains in $\mathcal{B}_\epsilon$ at all times, which proves stability.    
\end{proof}

\subsection{Semi-definite programming and SOS polynomials} \label{sec:SOS}
According to Theorem \ref{thm:Lyapunov1}, proving Lyapunov stability with the aid of Lyapunov's direct method requires the construction of functions, which are provably non-negative under certain inequality constraints. In general there is no efficient computational method to test non-negativity of a function or even of a polynomial function \cite{murty1985some}. 

However for a polynomial function to be non-negative, it is sufficient to prove that it can be written as the sum of squares of some polynomials, i.e. it is SOS.  Testing if a polynomial is SOS can be formulated as a semi-definite programming task, for which efficient numerical implementations exist \cite{marshall2008positive}. The theory of SOS polynomials has also been extended to optimization algorithms over SOS polynomials \cite{parrilo2003semidefinite}. Hence, SOS polynomials are highly useful tools for the verification of Lyapunov stability. Polynomials, which are positive under polynomial equality or inequality constraints can also be constructed by searching for unconstrained SOS polynomials, see Table~ \ref{tab:SOS}. A combination of equality and inequality constraints can also be treated in  a similar fashion, but this straightforward extension is not shown in the table.

\begin{table}
\centering{
\begin{tabular}{|p{0.4\columnwidth}|p{0.4\columnwidth}|}\hline
\bf{Condition to be satisfied}  & \bf{Sufficient condition using SOS polynomials}\\\hline
Find  $f(x):\mathbb{R}^n\rightarrow \mathbb{R}$ such that $\forall x\in \mathbb{R}^n: f(x)\geq 0$  
& 
Find  $f(x)\in \mathcal{P}_{SOS}$ 
\\ \hline
Given  $g_i(x)\in\mathcal{P}:\mathbb{R}^n\rightarrow \mathbb{R}$ ($i=1,2,...,m$), find  $f(x):\mathbb{R}^n\rightarrow \mathbb{R}$ such that $\forall i: g_i(x)=0 \implies f(x)\geq 0$  
& 
Find  $h_i(x)\in\mathcal{P}$ ($i=1,2,...,m$) such that $f(x)-\sum_{i=1}^mh_i(x)g_i(x)\in \mathcal{P}_{SOS}$ 
\\ \hline
Given $g_i(x)\in\mathcal{P}:\mathbb{R}^n\rightarrow \mathbb{R}$ ($i=1,2,...,m$), find  $f(x):\mathbb{R}^n\rightarrow \mathbb{R}$ such that $g(x)\geq 0 \implies f(x)\geq 0$  
& 
Find  $h_i(x)\in \mathcal{P}_{SOS}$ ($i=1,2,...,m$) such that $f(x)-\sum_{i=1}^mh_i(x)g_i(x)\in \mathcal{P}_{SOS}$
\\ \hline 
\end{tabular}
\caption{Sufficient conditions of (constrained) non-negativity  using SOS polynomials. $\mathcal{P}$, and $\mathcal{P}_{SOS}$ denote the set of polynomials, and SOS polynomials.}
\label{tab:SOS}
}
\end{table}


In a recent paper \cite{posa2015stability}, Posa et al., point out that the admissibility and consistency constraints of impacting systems are often polynomial functions of appropriately chosen variables. Hence Lyapunov's direct method can be implemented as an SOS optimization problem. Even though the replacement of non-negativity by the SOS property is conservative, the Lyapunov stability of several systems with one single contact point was succesfully verified by \cite{posa2015stability}. The same method could also provide conservative estimates of the basin of attraction of a stable equilibrium as well as proof of positive invariance of some subsets of $\mathcal{S}$, however we are not interested in these extensions in the present paper. Attempts to test the stability of systems with multiple contacts have not been reported.

\section{Extended stability theory}\label{sec:extended}


In order to test applicability to multi-contact systems, 
the stability test proposed by \cite{posa2015stability} has been implemented using the Mosek solver \cite{andersen2000mosek} available through the Yalmip toolbox of MatLab \cite{lofberg2004yalmip}. We tested the equilibria of the previously introduced planar biped (Figure \ref{fig:body}). Albeit this system has a highly non-trivial behaviour, almost exact conditions of stability are available if impacts are inelastic \cite{varkonyi2017lyapunov,or2021experimental}. Little is known about stability in the more general case of $e>0$, however highly conservative stability conditions have been reported in \cite{or2008hybrid,varkonyi2012lyapunov}. These preliminary tests proved unsuccessful: no certificates of stability were found by the algorithm, moreover the corresponding semi-definite programming task was in most cases flagged by the Mosek solver as provably infeasible. This negative result suggests that the original stability test has limited ability to verify stability in multi-contact systems. The failure of the test inspired the development of an improved stability theory, which can also address these systems.

In a previous work of the author \cite{varkonyi2012lyapunov}, a highly conservative condition of stability was developed for the planar biped in the case of $\phi_1=\phi_2=0$. That work used a Lyapunov-type function composed as the envelope of 2 smooth functions: $V(x)=\max(V_1(x),V_2(x))$. The first one was the total mechanical energy of the system, and $V_2(x)$ was found by trial and error. Analytical conditions were derived under which $V_2$ was decreasing over time in each contact mode except FF. For the FF mode it was proved that $V_2$ may increase, however each episode of free flight is followed by an impact at which $V_2$ drops to a lower value such that the net change of $V_2$ is negative. Those properties could be used to prove Lyapunov stability. In the present paper, these ideas are combined with the results of \cite{posa2015stability}, in order to develop a less conservative, algorithmic stability test in which manual search by trial-and-error is replaced by algorithmic search using a semi-definite program. First, an improved stability theory is developed following the ideas outlined above.  Throughout the rest of the paper the ZOD approximation of the dynamics is considered.

\subsection{The Lyapunov function $V(x)$ may increase temporarily along paths}

Lyapunov's direct method requires that $V(x)$ is a strictly non-increasing function along solution trajectories and along jumps associated with impacts. This condition can be replaced by a weaker condition allowing $V(x)$ to increase temporarily along solution trajectories, provided that it has a non-increasing overall trend. In particular, we will search for a countable set of times $t_i$, $i=1,2,...$, such that the values of $V(x(t_i^+))$ form a strictly non-increasing sequence, furthermore the time-differences $t_{i+1}-t_i$ are sufficiently small to prevent the escape of trajectories from the proximity of the equilibrium state during the time interval $(t_{i},t_{i+1})$. In particular, the discrete set of times $t_{i}$ will be chosen as times of impact. This choice is motivated by the observation, that successful Lyapunov functions are often closely related to kinetic energy, which is known to decrease during every impact. 

In order to develop the idea outlined above into a formal statement, one needs to establish upper bounds of time differences between contact mode transitions. 
%
First we introduce the concept:
\begin{definition}\emph{
A system is called impact-bounded if there exists a strictly increasing function $\zeta(\xi): \mathbb{R}\rightarrow\mathbb{R}$ with $\zeta(0)=0$ such that for any contact mode $\mathcal{M}$ and for any state $x\in\mathcal{S}$ in which mode $\mathcal{M}$ is admissible and consistent, the next contact mode transition or impact along the trajectory through $x$ occurs no later than at time $t\leq\zeta(|x|)$.
\label{def:bounded}
}\end{definition}
Then, two useful results are formulated. The first one ensures that contact mode transitions and impacts occur frequently
\begin{lemma}\emph{
Every system is impact-bounded under the ZOD at an unambiguous equilibrium $x_0=0$.
\label{lem:bounded}
}\end{lemma}
\begin{proof}
Ambiguity means that every non-static contact mode $\mathcal{M}$ is inconsistent in the equilibrium state $x_0$. It is straightforward to see that for an unambiguous equilibrium, at least one of the following properties must hold for each mode $\mathcal{M}$ under the ZOD:
\begin{itemize}
\item at least one of the admissibility constraints, or one of the consistency constraints referring to contact forces is violated. In this case, $\mathcal{M}$ may not be realized at all. 
\item at least for one contact $i$, the acceleration of the system in mode $\mathcal{M}$ implies $\dot{\gamma_i}< 0$. 
\item at least for one contact $i$, the sign of $\sigma_i$ dictated by the admissibility constraints and the sign of $\dot{\sigma_i}$ corresponding to the acceleration of the system in mode $\mathcal{M}$ satisfy $\sigma_i\dot{\sigma_i}< 0$ 
\end{itemize}
In the first case, there is no need to investgate impact times in mode $\mathcal{M}$. In the second and third cases,  either the normal velocity $\dot{\gamma_i}$ or the tangential velocity $\sigma_i$ approaches 0 at a constant rate. Hence one of the admissibility constraints referring to the signs $\gamma_i$, $\dot{\gamma_i}$ or  $\sigma_i$ is violated after bounded time.
This finding implies the statement of the theorem. The interested reader may find explicit expression of time bounds for a system with two contact points in \cite{varkonyi2017lyapunov}.
\end{proof}

%
The second statement establishes bounds for trajectories completed in short time intervals: 
\begin{lemma}\emph{
There exists a strictly increasing function $\eta(\xi): \mathbb{R}\rightarrow\mathbb{R}$ with $\eta(0)=0$ 
such that for any trajectory $x(t)$ under the ZOD, which includes no impacts  within a time interval $(t_1,t_2)$ and $|x(t_1)|<1$ the following bound is satisfied:
$|x(t_2)|
\leq |x(t_1)|+\eta(t_2-t_1)$
under the ZOD.
\label{lem:bound-noimpact}
}\end{lemma}

\begin{proof}
In each contact mode $\mathcal{M}$ the system has a constant acceleration $\dot{v}_\mathcal{M}$ under the ZOD. Let $a=\max_\mathcal{M}|\dot{v}_\mathcal{M}|$ denote the largest absolute value among those accelerations.

By using kinematics of constantly accelerating motion, the following bounds are obtained:
\begin{align}
|v(t_2)-v(t_1)|&\leq a(t_2-t_1)
\\
|q(t_2)-q(t_1)|&\leq |v(t_1)|(t_2-t_1)+\frac{1}{2}a(t_2-t_1)^2\\
&\leq 1\cdot (t_2-t_1)+\frac{1}{2}a(t_2-t_1)^2
\end{align}
which implies
$$
|x(t_2)-x_0|\leq |x(t_1)-x_0|+(a+1)(t_2-t_1)+\frac{1}{2}a(t_2-t_1)^2.
$$
proving the statement.
\end{proof}

As we will see in Sec. \ref{sec:extended}.C, these results together enable us to prove stability by constructing a function $V(x)$, which is not monotonic, but whose values evaluated at times of impacts and contact mode transitions form a decreasing sequence.

%
%

\subsection{Stability certificates via multiple Lyapunov functions }

The application of the original stability condition relies on finding a Lyapunov function $V(x)$ with the properties specified in Theorem \ref{thm:Lyapunov1}. Stability can also be verified if we find a finite set of functions $V_i(x)$ ($i=1,2,...,n$), none of which satisfies condition 3a) of Theorem \ref{thm:Lyapunov1} provided that a related condition 
\begin{equation}
\max_i V_i(x)\geq\alpha(|x-x_0|)
\label{eq:V_i}
\end{equation} 
is satisfied, furthermore each one of the functions satisfies condition 3b).

This extension may appear somewhat superfluous, because the existence of a set of function with these properties implies that the non-smooth function defined by $V(x)=\max_iV_i(x)$ satisfies the original conditions 3a-b) of the Theorem. Nevertheless stability tests using SOS programming can benefit from such an extension, as the upper envelope of two polynomial functions is usually not a polynomial. Hence the use of several Lyapunov functions makes a wider set of Lyapunov candidates available for the solver.  

Notably, a second Lyapunov function can be used to inform the solver about the dissipative nature of frictional dynamics. In particular one can choose $V_1(x)$ as the total mechanical energy of the system and search for a second function $V_2(x)$, which together satisfy \eqref{eq:V_i}. 

\subsection{An improved stability condition}

Now we are ready to state the main result of the paper:
\begin{theorem}\emph{ \label{thm:Lyapunov3} Let $x_0=(0,0)$ be an unambiguous equilibrium state. If 
there exists
\begin{enumerate}
\item a positive value $h$
\item a continuous, strictly increasing scalar functions $\alpha$ with $\alpha(0)=0$
\item a finite set of continuously differentiable functions $V_i : \mathbb{R}^{2n}\rightarrow \mathbb{R}$ ($i=1,2,...,\upsilon$) such that for any initial state $x\in \mathcal{B}_h \cap \textrm{cl}(\mathcal{A})$, 
\begin{enumerate}
\item 
$\max_iV_i(x)\geq\alpha(|x-x_0|)$
\item $V(x)\geq V(x^*)$ where $x^*$ denotes the next point of contact mode transition  or the next post-impact state along the the trajectory through $x$, whichever occurs earlier.
\end{enumerate}
\end{enumerate}
then $x_0$ is stable in the sense of Lyapunov.
}\end{theorem}

Clearly, Theorem \ref{thm:Lyapunov3} is a generalization of Theorem \ref{thm:Lyapunov1}. The latter is recovered if one requires $\upsilon=1$, and condition 3) of Theorem \ref{thm:Lyapunov3} is considered in the limit $t_2\searrow t_1$.

The proof of Theorem \ref{thm:Lyapunov3} follows the logical steps of Theorem \ref{thm:Lyapunov1}. As it has been pointed out, the relaxed condition 3a) in the statement does not require any significant change in the proof, as the non-smooth function $\max_iV_i(x)$ satisfies the requirements of the original theorem. The relaxed condition 3b) requires some adaptation of the proof.

\begin{proof}
Consider an arbitrary $\epsilon>0$. Define $c$ as $c=\alpha(\min(\epsilon,h))$. Then $V(x)\geq c$ along the boundary of $\mathcal{B}_{\min(h,\epsilon)}$ hence $V(x)$ has a closed, connected sublevel set  $\Omega_c:V(x)\leq c$ such that $\Omega_c$ contains $x_0$ in its interior, furthermore $\Omega_c\subset \mathcal{B}_{\min(h,\epsilon)}$. Next we can find a scalar $\beta>0$ such that $\mathcal{B}_\beta\subset(\Omega_c\cup\mathcal{A}^c)$ where $^c$ means complement.

Since $\eta$ and $\zeta$ are strictly monotonic functions, there is a unique positive value of $\delta$ satisfying the equation 
\begin{equation}
\delta+\eta(\zeta(\delta)))=\beta
\label{eq:deltabeta}
\end{equation}
If the initial point of a trajectory at time $t_1$ is in $\mathcal{B}_\delta$, then by Lemma \ref{lem:bounded} the next contact mode switch or impact occurs no later than at time $t_2=t_1+\zeta(\delta))$. According to Lemma \ref{lem:bound-noimpact}, and equation \eqref{eq:deltabeta}, the state of the object remains within the bound given by $|x(t)-x_0| \leq \beta$ during the time interval $t\in (t_1,t_2)$, which is a subset of $\mathcal{B}_\epsilon$. By using Condition 3b), we also know that $|x(t_2^+)| \leq |x(t_1)|$, i.e. the state $|x(t_2^+)$ is also inside $\mathcal{B}_\delta$ just like the initial state $x(t_1)$. This observation allows the recursive application of the previous arguments for each phase of motion free of impacts and contact mode transitions. This way we arrive to the final conclusion that the state of the systems remains within $\mathcal{B}_\epsilon$ at all times,  implying Lyapunov stability.  
\end{proof}

\section{STABILITY OF A PLANAR BIPED}
\label{sec:example}
In this section we apply the extended stability test to a planar model of a rigid body with 2 contact points, which we will refer to in the sequel as 'biped'. 

\subsection{Implementation of the stability condition}



The planar biped has 16 contact modes, which are two-letter words composed of the set of letters $\lbrace P,N,F,S\rbrace$. Among these, exactly 10 are kinematically admissible in some parts of state space. For example if $\cos\phi_1,\cos\phi_2>0$. then modes $PN$, $NP$, $PS$, $SP$, $SN$, $NS$ are kinematically inadmissible for all states of the object.
The remaining 10 modes may be admissible or inadmissible depending on state. For each of them, the values of contact forces and instantaneous accelerations of the system are state-independent under the ZOD, and thus they can be calculated in advance. For example the $FF$ contact mode there are no contact forces and the acceleration is given by the value of $C$, see equation \eqref{eq:C}.


After the preliminary steps outlined above, the conditions of Theorem \ref{thm:Lyapunov3} are formulated as an SOS program and tested. Such an implementation requires
\begin{itemize}
\item specification of $h$ in condition 1). We use $h=1$.
\item Lyapunov function candidates. We choose $V_1(x)$ to be total mechanical energy that is
\begin{equation}
V_1(x)=\frac{1}{2}(u^2+v^2+\omega^2+z\cos\alpha-x\sin\alpha+T\theta) \label{eq:V1}
\end{equation}
$V_2(x)$ is chosen as 
\begin{equation}
V_2(x) = x + P_1^D(\phi,z,u,v,\omega)
\end{equation}
where $P_1^D$ denotes a polynomial containing all terms with degrees within the interval $(1,D)$. $V_2$ does not have constant terms due to the required property $V(x_0)=0$. We chose $D=3$ as larger values of $D$ result in too many unknown coefficients, and high computational cost. In addition, the nonlinear terms in the variable $x$ are omitted in order to reduce the number of unknown coefficients. This choice is motivated by the fact that $x$ is a cyclic coordinate of the ZOD dynamics. With the restrictions above $V_2(x)$ has altogether 55 unspecified coefficients to be determined.
\item existence of an appropriate function $\alpha(\xi)$ satisfying condition 3a). In the Appendix, we prove that such a function exists if 
\begin{align}
V_2(x)&\geq x+u^4+v^4+\omega^4 \label{eq:V2bound} \\
0<\alpha&<\pi/2 \label{eq:alphalim}\\
L_1L_2&<0 \label{eq:L1L2}
\end{align}
In what follows, we restrict our attention to systems satisfying \eqref{eq:alphalim}, \eqref{eq:L1L2}. The polynomial inequality \eqref{eq:V2bound} is added to the non-negativity conditions to be fulfilled. 
\item Implementation of condition 3b) for all possible forms of motion. The function $V_1$ is non-increasing, hence all we need is to prove monotonicity of $V_2$. If $e=0$, then the motion does not have free flight phase. In this case, we use the following conditions:
\begin{itemize}
\item $\frac{d}{dt}V_2(x(t)) \leq 0$ for all of the non-static contact modes involving an active contacts $\lbrace PF,SF,NF,FP,FS,FN,PP,NN \rbrace$, and for all states $x\in \mathcal{B}_h \cap \textrm{cl}(\mathcal{A})$ where that contact mode is admissible and consistent
\item  $V_2(M(x(t^-)))-V_2(x(t^-))\leq 0$ for all possible pre-impact states $x(t^-)\in \mathcal{B}_h \cap \textrm{cl}(\mathcal{A})$ where $M$ is the impact map.
\end{itemize} 
If $e>0$ then free-flight is also possible. Free-flight is always followed by an impact, and we use the following condition: 
\begin{itemize}
\item  $V(M(x(t^-)))-V(t-\Delta)\leq 0$ for all possible pre-impact states $x(t^-)$ and for all possible durations $\Delta$ of a free-flight phase immediately before that impact. As before, $M$ means the impact map.
\end{itemize} 
\end{itemize}
All of these conditions as well as the corresponding conditions of admissibility and consistency are given in closed form in the Appendix. All of them are polynomial equations or inequalities in an appropriate set of variables (which includes state variables as well as some additional variables such as $\Delta$ in the case of free flight). Hence the whole set of conditions can be recast as an SOS programming task as demonstrated in Subsection \ref{sec:SOS}.

The extended stability test has been implemented in Yalmip optimization toolbox. In the following section, some test results are presented.

\subsection{Regions of stability in parameter space}

First, the stability test is tested and verified using the example of a biped on a slope of angle $30^\circ$ (that is $\phi_1=\phi_2=0$, $\alpha=30^\circ$) subject to gravitational load ($T=0$, $F=1$) with inelastic impacts ($e_1=e_2=0$). The friction coefficient $\mu_2$ is kept fixed while $\mu_1$, and the distance $2d$ between the two contact points are varied systematically. The positions of the contact points are chosen as $h=1$, $L_1=h\tan\alpha+d$, $L_2=h\tan\alpha-d$. Note that this choice of $L_1$, and $L_2$ ensures that the object does not topple in $SF$ or $FS$ mode in a static state. 

This system has been investigated previously by several works. If $\mu_2>\tan\alpha$ then the object is in frictional equilibrium for all values of $\mu_1$ and $d$. Nevertheless the equilibrium is trivially unstable due to ambiguity \cite{or2008hybrid-I} if the following two conditions are satisfied:
\begin{align}
d&<h\tan\alpha \\
\mu_1&<\tan\alpha-d/h \label{eq:mu1lim}
\end{align}
In this case, the PF mode is consistent in addition to SS, i.e. a microscopic perturbation of the equilibrium state may initiate downhill slip motion at point 1 accompanied by separation of point 2.

Even when $\mu_1$ is above the limit \eqref{eq:mu1lim}, the object may be unstable due to self-excited inverse chatter motion as first pointed out by \cite{varkonyi2012lyapunov}. The exact range of parameter values corresponding to instability has been determined by semi-analytical investigation of an appropriately defined Poincar\'e map \cite{varkonyi2017lyapunov,or2021experimental}. The newly proposed Lyapunov stability test is able to verify stability in a significant portion of the stable region of model parameters (Fig. \ref{fig:chart1}). This test indicates that the new test is conservative, but it is successful in some highly non-trivial cases. 

The exact conditions of Lyapunov stability are available for the previous example, which possesses two point contacts with ideally inelastic impacts unknown. In fact, to the author's best knowledge this is the only known case of rigid (multi-)body systems with multiple, unilateral frictional contacts, where conditions of Lyapunov stability are known. 

As second example we choose the same model with slope angle $\alpha=20^\circ$, partially elastic impacts $e_1=e_2=0.1$ and equal coefficients of friction $\mu_1=\mu_2$ at the two contact points. 

This system is in frictional equilibrium if $\mu_1,\mu_2\geq \tan20^\circ\approx 0.36$ and the equilibrium is unambiguous in all cases.
The exact parameter range implying Lyapunov stability is unknown, however this example has been investigated analytically by \cite{varkonyi2012lyapunov} with the aid of a manually constructed Lyapunov-type function chosen based on physical intuition. That function could verify stability in a range of model parameters illustrated by Fig. \ref{fig:chart2}.

We test the stability using the newly proposed  SOS program for various combinations of the parameters $d$ and $\mu_1=\mu_2$. The test verified stability in a significantly larger range than the region of stability, verified by \cite{varkonyi2012lyapunov}.


\begin{figure}[h]
\centering
\includegraphics[width=1.1\columnwidth]{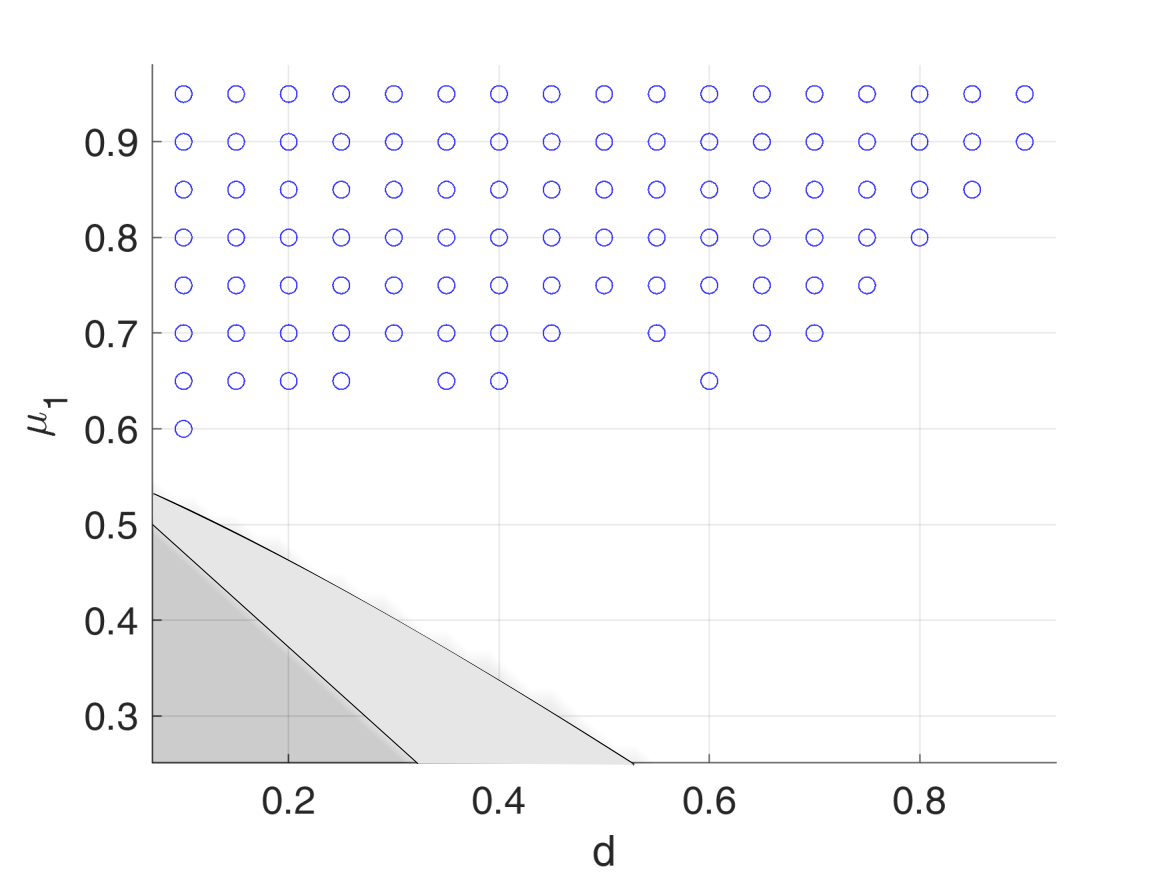}
\caption{Equilibrium and stability charts of a planar biped with $e_1=e_2=0$,  $T=0$, $F=1$, $\alpha=30^\circ$, $h=1$, $\phi_1=\phi_2=0$, and $L_1,L_2=h\tan\alpha\pm d$, $\mu_2=2$ as a function of $\mu_1$ and $d$. Background colors represent analytical results: unstable equilibrium due to ambiguity (dark), unstable equilibrium due to inverse chatter (middle) and stable equilibrium (white). Circles indicate stability region verified by the SOS program.}
\label{fig:chart1}
\end{figure}
\begin{figure}[h]
\centering
\includegraphics[width=1.1\columnwidth]{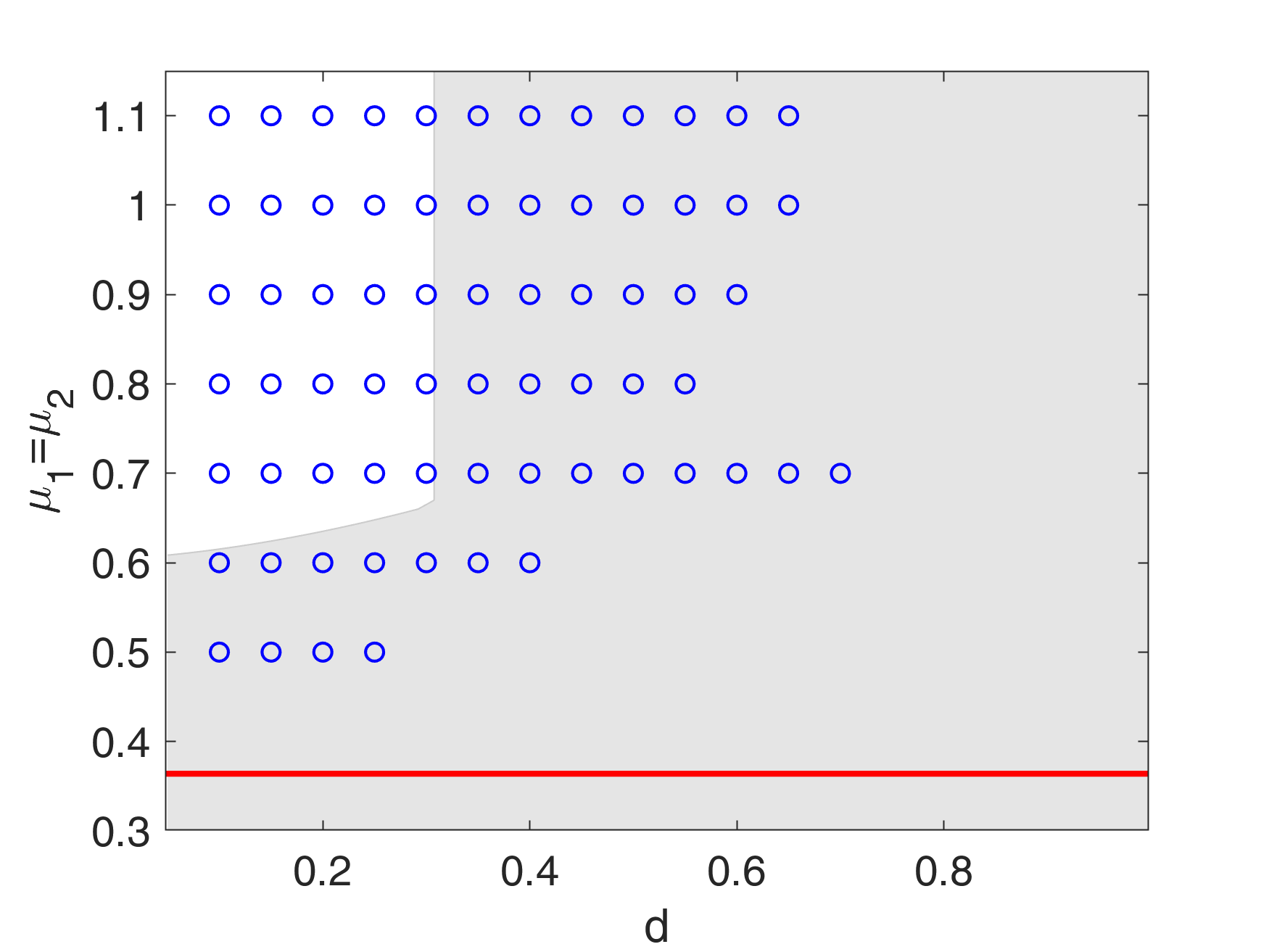}
\caption{Equilibrium and stability charts of a planar biped with $e_1=e_2=0.1$, $T=0$, $F=1$, $\alpha=20^\circ$, $h=1$, $\phi_1=\phi_2=0$, and $L_1,L_2=h\tan\alpha\pm d$, as a function of $\mu_1=\mu_2$ and $d$. Solid line indicates minimum value of friction coefficient for frictional equilibrium. White background  corresponds to region of stability succesfully  verified by \cite{varkonyi2012lyapunov}. Circles indicate larger stability region verified by the SOS program.}
\label{fig:chart2}
\end{figure}


\section{CONCLUSIONS}

Lyapunov stability analysis of rigid multibody systems with unilateral, frictional contacts is a fundamental question of robotics and object manipulation. Recent theoretical and experimental works demonstrated the existence of various types of instability in response to small perturbations initiating liftoff at the contacts. Despite initial attempts to predict stability of some simple model systems, general stability tests remain unavailable. The complexity of contact-induced rigid body dynamics makes analytical investigation infeasible even for seemingly simple systems. Algorithmic stability tests have great potential in more complex cases. An important step in this respect was made by \cite{posa2015stability}, where the compatibility of SOS programming with the special version of Lyapunov's direct method applicable to hybrid dynamical systems induced by rigid contact was recognized. Despite the success of this method in several related problems, it fails to verify Lyapunov stability of simple multi-contact test problems. In the present work, several improvements of the method have been proposed, which enabled successful verification of Lyapunov stability in those cases. It has been found for a particular example that the algorithmic test outcompeted a previously proposed manual stability test based on physical intuition. At the same time, the stability test appears to be conservative in those cases where the exact condition of stability are known.

The author foresees that future developments will radically extend the applicability of algorithmic stability tests to contact-induced dynamics. A fundamental challenge to be solved is reduction of high complexity of the algorithm (currently exponential in the number of contact points), which emanates from the complexity of the hybrid dynamics approach to contact problems. 

\section*{APPENDIX}

\subsection{Verification of condition 3a) in Theorem \ref{thm:Lyapunov3}}

The non-penetration constraints under the ZOD \eqref{eq:gammabipedZOD} now yields $\gamma_i = z_i+L_i\theta\geq 0$ for $i=1,2$ since $\phi_i=0$. This relation is combined  with \eqref{eq:L1L2} to obtain the bound
\begin{align}
z\geq \min_i|L_i|\cdot |\theta| \label{eq:|phi|bound}
\end{align}
which will be used in the sequel.

Now we separate state space to two regions along the plane $y = \frac{1}{2}\cot\alpha z$ and the following bounds are established.

If $y \geq \frac{1}{2}\cot\alpha z$ then this inequality along with \eqref{eq:V2bound}, \eqref{eq:|phi|bound} implies
\begin{align}
V_2(x)&\geq y = |y|\\
      &\geq \frac{1}{2}\cot\alpha\cdot |z|\\
      &\geq \frac{1}{2}\cot\alpha\min_i|L_i|\cdot|\theta| \\
V_2(x)&\geq u^4,v^4,\omega^4
\end{align}
On the other hand, if If $y \leq \frac{1}{2}\cot\alpha z$ then this inequality and  \eqref{eq:V1}, \eqref{eq:|phi|bound} imply
\begin{align}
V_1(x)\geq z\cos\alpha-y\sin\alpha &\geq\frac{1}{2}\cos\alpha \cdot z\\
 &\geq \frac{1}{2}\cos\alpha \min_i|L_i|\cdot|\theta|\\
    V_1(x)  \geq z\cos\alpha-y\sin\alpha &\geq \sin\alpha \cdot |y|\\
V_1(x)\geq \frac{1}{2}u^2,\frac{1}{2}v^2,\frac{1}{2}\omega^2
\end{align}

In both cases, we have found for each one of the state variables $\sigma\in\lbrace{y,z,\theta,u,v,\omega\rbrace}$ an associated strictly increasing function $\alpha_\sigma$ satisfying $\alpha_\sigma(0)=0$ such that $\max_iV_i(x)\geq \alpha_\sigma(|\sigma|)$. It is straightforward to show that his property implies condition 3a.  

\subsection{State-dependent conditions associated with contact modes}
We present state-dependent admissibility and consistency condition of each contact mode. All of these conditions are formulated as polynomial inequalities in an appropriate set of variables.

\textbf{PF, NF, FP, FN modes:}
Let $i$ denote the active contact, and $j=3-i$ the inactive one. The normal contact force is determined from the consistency constraint $\ddot{\gamma}_i=0$ using the equation of motion \eqref{eq: manipulator} and the kinematic relations given in Sec. \ref{sec:ZOD}:
$$
\lambda_{\pm,i}=\frac{\cos\alpha}{1+L_i^2\mp\mu_iL_iH}
$$
where the $\pm$ signs should be understood as $+$ for slip in the positive direction and $-$ for negative slip. If $\lambda_{\pm,i}<0$ then the corresponding mode is not consistent in any state, hence the condition given below can be disregarded. In the opposite case, we express
the time derivative of the vector $x$ of state veriables as
\begin{align}
\dot x_{\pm,i} &= (u,v,\omega,\dot u,\dot v,\dot\omega)\\
&=
 (u,v,\omega,
\sin\alpha\mp\mu_i\lambda_{\pm,i},
-\cos\alpha+\lambda_{\pm,i},
(L_i\mp\mu_iH)\lambda_{\pm,i}
)
\label{eq:stabconslip}
\end{align}
Finally, condition 3b) is given by
\begin{align}
-J_V(x)\dot x_{\pm,i}\geq 0 &
\textrm{ whenever}
&\gamma_i(q)=0\\
&& \dot \gamma_i(q)=0\\
&&\gamma_j(q)\geq 0\\
&&\dot\sigma_i(q,v)
\begin{array}{c}\geq\\ \leq\end{array} 
0\\
&&1-x^Tx \geq 0
\end{align}
where $J_V(x)$ is the Jacobian of the function $V(x)$, which can be expressed in terms of the 55 unknown coefficients of $V(x)$.

\textbf{SF, FS modes:}
Let $\lambda$ and $\lambda_t$ denote normal and tangential contact forces at the sticking contact $i$, and let $j$ denote the other contact. Then, the contact forces can be determined using the kinematic constraints $\ddot\gamma_i=\dot\sigma_i=0$:
$$
[\lambda,\lambda_t]^T=-P[-\cos\alpha,\sin\alpha]^T
$$
with 
$$
P=\left[\begin{array}{cc}
1+L_i^2&L_iH\\
L_iH&1+H^2
\end{array}\right]
$$
If $\lambda<0$ or $|\lambda_t|>\mu_i\lambda$, then the contact mode under investigation can be omited otherwise the time derivative of $x$ is
\begin{align}
\dot x_{S,i} &= 
 (u,v,\omega,
\sin\alpha+\lambda_t,
-\cos\alpha+\lambda,
L_i\lambda+H\lambda_t
\end{align}
and the stability conditions are formulated as
\begin{align}
-J_V(x)\dot x_{S,i}\geq 0 &
\textrm{ whenever}
&\gamma_i(q)=0\\
&& \dot \gamma_i(q)= 0\\
&&\gamma_j(q)\geq 0\\
&&1-x^Tx \geq 0
\label{eq:stabconstick}
\end{align}
We note that $\dot x_{S,i}$ is a convex combination of $\dot x_{+,i}$ and $\dot x_{-,i}$, i.e. $\dot x_{S,i} = \kappa \dot x_{+,i} + (1-\kappa) \dot x_{-,i}$ for some $0\leq \kappa\leq 1$. Hence the conditions \eqref{eq:stabconslip} for positive and negative slip together imply \eqref{eq:stabconstick}.

\textbf{PP, NN modes:}
In this case, the two normal contact forces are  determined by the constraints $\ddot \gamma_1=\ddot \gamma_2=0$:
$$
[\lambda_1,\lambda_2]^T=-P^{-1}[\cos\alpha,\cos\alpha]^T
$$
where
$$
P=\left[\begin{array}{cc}
1+L_1^2\mp\mu_1HL_1&1+L_1L_2\mp\mu_2HL_1\\
1+L_1L_2\mp\mu_1HL_2&1+L_2^2\mp\mu_2HL_2
\end{array}\right]
$$
If any of the two contact force values is negative, than the conditions associated with the contact mode under investigation can be disregarded. In the opposite case, we have 
\begin{align}
\dot x_{\pm} 
&=
 (u,v,\omega,
\sin\alpha\mp\sum_i\mu_i\lambda_i,
-\cos\alpha+\sum_i\lambda_i,
\sum_i(L_i\mp\mu_iH)\lambda_i)
\end{align}
and the stability conditions are formulated as
\begin{align}
-J_V(x)\dot x_{\pm} \geq 0 &
\textrm{ whenever}
&\gamma_1(q)=\gamma_2(q)=0\\
&& \dot \gamma_i(q)=\dot\gamma_2 = 0\\
&&\sigma_1(q,v)
\begin{array}{c}\geq\\ \leq\end{array} 
0\\
&&1-x^Tx \geq 0
\end{align}

\textbf{Impacts:} 
The impact maps of sticking and slipping impacts have been given in closed form in the main text. Instead of treating the three types of impacts separately, slipping and sticking impacts  are addressed now in a common framework as follows. Let $i$ be the contact point undergoing an impact. Let $\Lambda_{N,i}$ and $\Lambda_{T,i}$ denote net impulses in the normal and tangential directions during that impact.

Then $\Lambda_{N,i}$ can be expressed in terms of the pre-impact state and of $\Lambda_{T,i}$ by using the condition $\dot\gamma_i(x(t^+))=-e\dot\gamma_i(x(t^-))$ and \eqref{eq:manipulator2}:
$$
\Lambda_{N,i}=\frac{-(1+e)(v_z(t^-)+l_i\omega(t^-)) + \Lambda_{T,i} hl_i}{1+l_i^2}
$$
Then, the post-impact state becomes
\begin{align}
x^+ &= x^- + 
 (0,0,0,\Lambda_{T,i},\Lambda_{N,i},\Lambda_{T,i} H+\Lambda_{N,i} L_i)
\end{align}
and the stability conditions are formulated as
\begin{align}
V(x^-)-V(x^+) \geq 0 &
\textrm{ whenever}
&\gamma_i(q)=0\\
&& \dot \gamma_i(x^-)\leq 0\\
&&\gamma_j(q)\geq 0\\
&&(u^-+\omega^+H)\lambda_t\leq 0\\
&&\mu_i^2\lambda^2-\lambda_t^2\geq 0\\
&(\mu_i^2\lambda^2-\lambda_t^2)&\cdot(u^-+\omega^+H)=0
\end{align}
The last three conditions express Coulomb's law for the impact. The conditions listed above are polynomial in the set of variables $\lbrace y,z,\theta,u,v,\omega,\lambda_t\rbrace$

\textbf{FF mode followed by an impact:}

\addtolength{\textheight}{-12cm}   





\bibliographystyle{ieeetr}
\bibliography{refs}

\end{document}